\documentclass[onecolumn, draftclsnofoot,11pt]{IEEEtran}

\IEEEoverridecommandlockouts

\usepackage{mathrsfs}
\usepackage{tikz}
\usepackage[utf8]{inputenc}
\usepackage{amsfonts} 
\usepackage{epstopdf}
\usepackage{graphicx}
\usepackage{bbm}
\usepackage{enumerate}
\usepackage{epsf,epsfig,verbatim,amssymb,amsmath,array,cite,amsthm,subfigure,multicol,multirow}  
\usepackage{psfrag,xspace}

\usepackage{bm}
\usepackage{hhline}
\usepackage{mathabx}
\usepackage{physics}
\usepackage{xcolor}
\definecolor{darkblue}{rgb}{0.1,0.1,0.8}
\definecolor{brickred}{rgb}{0.8, 0.25, 0.33}
\definecolor{DarkGreen}{rgb}{0,0.6,0}

\newtheorem{theorem}{Theorem}

\newtheorem{lemma}{Lemma}
\newtheorem{corollary}{Corollary}

\newtheorem{proposition}{Proposition}

\theoremstyle{definition}
\newtheorem{definition}{Definition}

\newtheorem*{prob*}{Problem}

\newtheorem{remark}{Remark}

\allowdisplaybreaks 
\makeatletter
\def\old@comma{,}
\catcode`\,=13
\def,{%
  \ifmmode%
    \old@comma\discretionary{}{}{}%
  \else%
    \old@comma%
  \fi%
}
\makeatother

\usepackage{etoolbox}
\providetoggle{Blue_revision}
\settoggle{Blue_revision}{true}

 \usepackage{hyperref}
\hypersetup{
colorlinks=true, %
 pdfstartview={FitH},
    linkcolor=red,
    citecolor=blue, 
    urlcolor={blue!80!black}
}

\usepackage{graphicx}

\global\long\def\EE{\mathbb{E}}
\global\long\def\PP{\mathbb{P}}
\global\long\def\FF{\mathbb{F}}
\global\long\def\11{\mathbbm{1}}

\def\CalC{\mathcal{C}}
\def\CalD{\mathcal{D}}
\def\CalH{\mathcal{H}}

\def\CalN{\mathcal{N}}

\def\CalR{\mathcal{R}}

\newcommand{\Cal}[1]{\mathcal{#1}}

\global\long\def\+{\oplus}

\usepackage{stackengine}
\def\deq{\mathrel{\ensurestackMath{\stackon[1pt]{=}{\scriptstyle\Delta}}}}

\global\long\def\tensor{\otimes}
 \def\tr{\Tr}
 \def\id{\text{id}}

\allowdisplaybreaks

\def\TDelta{\mathcal{T}_{\delta}}

\def\PiA{\Pi_{\rho^A}}
\def\PiB{\Pi_{\rho^B}}

\def\PiuA{\Pi_{u^n}^A}
\def\PivB{\Pi_{v^n}^B}

\def\lambdauA{\lambda^A_{u^n}}
\def\lambdalA{\lambda^A_{l}}
\def\lambdavB{\lambda^B_{v^n}}
\def\lambdakB{\lambda^B_{k}}

\def\lambdalkAB{\lambda^{AB}_{l,k}}
\def\rhohatuA{\hat{\rho}^A_{u^n}}
\def\rhohatvB{\hat{\rho}^B_{v^n}}

\def\LambdauA{\Lambda^A_{u^n}}
\def\LambdalA{{\Lambda_l^A}}

\def\LambdakB{\Lambda^B_{k}}

\def\TDeltan{\mathcal{T}_{\delta}^{(n)}}
\def\IndiU1{\mathbbm{1}_{\left\{U^{n,(\mu_1)}(l)=u^n\right\}} }
\def\IndiV1{\mathbbm{1}_{\left\{V^{n,(\mu_2)}(k)=v^n\right\}} }

\def\rhotilduA{\tilde{\rho}_{u^n}^A}
\def\rhotildvB{\tilde{\rho}_{v^n}^B}

\def\sigmahatAu{\ket{\hat\sigma_{u^n}}^{AE}}
\def\sigmatildAu{\ket{\tilde\sigma_{u^n}}^{AE}}
\def\sigmahatBv{\ket{\hat\sigma_{v^n}}^{BF}}
\def\sigmatildBv{\ket{\tilde\sigma_{v^n}}^{BF}}
\def\sigmahatAl{{\hat\sigma_{l}}^{AE}}
\def\sigmatildAl{{\tilde\sigma_{l}}^{AE}}
\def\sigmahatBk{{\hat\sigma_{k}}^{BF}}
\def\sigmatildBk{{\tilde\sigma_{k}}^{BF}}

\def\PsiRhoNket{\ket{\Psi_{\rho^{\tensor n}}}^{ABCR}}
\def\clocc{$1$\mbox{-CLOCC}^\prime}
\newcommand{\ral}{\rangle}
\newcommand{\lal}{\langle}

\newcommand{\settildeR}[2]{#1 \in [2^{n\tilde{R}_{#2}}]}
\newcommand{\Unitary}[2]{U_{\scriptscriptstyle \!#1}^#2}

\usepackage{stmaryrd }
\usepackage{stackengine}

\usepackage{relsize}
\usepackage{upgreek}

\usepackage{acronym}
\usepackage{enumitem,kantlipsum}
\newacro{qc}[q-c]{quantum-to-classical }

\begin{document}
\pagenumbering{arabic}

\title{\huge Multi-Party Quantum Purity Distillation with Bounded Classical Communication}
\author{
\IEEEauthorblockN{  Touheed Anwar Atif and S. Sandeep
    Pradhan}\\
\IEEEauthorblockA{Department of Electrical Engineering and Computer Science,\\
University of Michigan, Ann Arbor, MI 48109, USA.\\
Email: \tt touheed@umich.edu, pradhanv@umich.edu}
}

\maketitle



\begin{abstract}
  We consider the task of distilling local purity from a noisy quantum state $\rho^{ABC}$, wherein we provide a protocol for three parties, Alice, Bob and Charlie, to distill local purity (at a rate $P$) from many independent copies of a given quantum state $\rho^{ABC}$. The three parties have access to their respective subsystems of $\rho^{ABC}$, and are provided with pure ancilla catalytically, i.e., with the promise of returning them unaltered after the end of the protocol. In addition, Alice and Bob can communicate with Charlie using a one-way multiple-access dephasing channel of link rates $R_1$ and $R_2$, respectively. The objective of the protocol is to minimize the usage of the dephasing channel (in terms of rates $R_1$ and $R_2$) while maximizing the asymptotic purity that  can be jointly distilled from $\rho^{ABC}$. To achieve this, we employ ideas from distributed measurement compression protocols, and in turn, characterize a set of sufficient conditions on $(P,R_1,R_2)$ in terms of quantum information theoretic quantities such that $P$ amount of purity can be distilled using rates $R_1$ and $R_2$. Finally, we also incorporate the technique of asymptotic algebraic structured coding, and provide a unified approach of characterizing the performance limits.
  
\end{abstract}

{\hypersetup{
colorlinks=true, %
 pdfstartview={FitH},
    linkcolor=black,
    citecolor=blue, 
    urlcolor={blue!80!black}
    pagenumbering= {roman}
}
}

\section{Introduction}
\label{sec:introduction}
A primary task in quantum information theory is to quantify the amount of local and non-local information present within a quantum information  source. For instance, the task of entanglement distillation aims at capturing the non-local correlations to transform a noisy shared state $\rho^{AB}$ into pure bell states (in particular, the ebit $|\Phi^+\ral$), in an aymptotic sense. A complementary notion to this task is the paradigm of local purity distillation, where pure ancilla qubits are distilled from a distributed state $\rho^{AB}$ using local unitary operations. 

Although it may seem unusual, local pure states cannot be considered as a free resource. One may argue that pure states can be obtained from a mixed state by performing a measurement, but this is only true after a measurement apparatus is initialized in a pure state. For this reason, the second law of thermodynamics  recognizes purity as indeed a resource \cite{alicki2004thermodynamics,horodecki2005local}. In this regard, the idea of distilling of local purity was first introduced in \cite{oppenheim2002thermodynamical,horodecki2003local} where the aim was to manipulate the qubits and concentrate the existing diluted form of purity. Two version of this problem have been introduced, (i) a single-party variant and (ii) a distributed version. In the former single-party scenario, also called as \emph{local purity concentration}, many copies of a noisy state $\rho^A$ are provided to Alice, and she aims at concentrating or extracting purity using only unitary operations. The authors in \cite{horodecki2003reversible} characterized the asymptotic performance limit of this protocol ($\kappa(\rho^A)$) as the difference between the number of qubits describing the system and the von Neumann entropy of the state $\rho^A$. For the latter case of distilling purity from a non-local distributed state, commonly termed as \emph{local purity distillation}, two parties, Alice and Bob, share many copies of the noisy state $\rho^{AB}$ and aim at jointly distilling pure ancilla qubits. Again, they are allowed to perform only local unitaries and but can communicate classically (LOCC), possibly through the use of a dephasing channel \cite{oppenheim2002thermodynamical}. Further, the protocols for both the variants require isolation (Closed-LOCC) from the environment which eliminated the possibility of unlimited consumption of the pure ancilla qubits. 
The authors in \cite{horodecki2003local} provided bounds for this problem in the one-way and the two-way classical communication scenarios. 

Later, Devetak in \cite{devetak2005distillation} considered a new paradigm called $\clocc$,
which was defined as an extension of Closed-LOCC, with (i) the allowance of using additional catalytic pure ancilla as long as these are returned back to the system, and (ii) the unlimited bidirectional classical communication replaced by unlimited one-way communication from Alice to Bob. Devetak obtained an information theoretic characterization of the distillable purity in the $\clocc$ setting (allowing additional catalysts) and highlighted its connection to the earlier known one-way distillable common randomness measure \cite{devetak2004distilling}. The usage of catalytic resource to improve the quantum information tasks was first introduced in \cite{JonathanEntanglement1999}. This further was extensively studied in a multitude of works, including but not limited to \cite{daftuar2001mathematical,van2002embezzling,turgut2007catalytic,aubrun2008catalytic,sanders2009necessary,brandao2015second,duarte2016self,bu2016catalytic,shiraishi2021quantum,lipka2021catalytic,ding2021amplifying,takagi2021correlation}. Building upon the work of \cite{devetak2005distillation}, the authors in \cite{krovi2007local} extended the result to a setting with bounded one-way classical communication, again allowing for the additional catalytic resource. They improved upon the classical communication rate by using  the Winter's approximate measurement \cite{winter}, instead of an $n$-letter product measurement, and extracted purity for the states obtained thereby. 

In this work, we revisit the task of distilling purity and consider a three-party setup. We ask the question of how many ancilla qubits can be distilled from a noisy state $\rho^{ABC}$, shared among three parties, Alice, Bob and Charlie. Similar to earlier problem formulation, we only allow local unitary operations at each party in a closed setting but permit the use of additional catalytic ancillas with the promise of returning them at the end of the protocol. In addition, similar to \cite{krovi2007local}, we only allow limited classical communication, which we model using a one-way multiple-access dephasing channel, with Alice and Bob as the senders and Charlie as the centralized receiver. 

The contributions of our work can be summarized as follows. We first formulate a three-party purity distillation problem, and develop a $\clocc$ multi-party purity distillation protocol for this problem capable of extracting purity from $n$ copies of the noisy shared state $\rho_{ABC}^{\tensor n}$, using only local unitary operations and a one-way multiple-access dephasing channel. Further, for $\rho_{ABC}^{\tensor n}$, we define the asymptotic performance limit of the problem as the set of all triples $(P,R_1,R_2)$, where $P$ denotes the amount of purity that can be distilled from $\rho^{ABC}$, using $R_1$ and $R_2$ bits of classical communication. 
Then we characterize a quantum-information theoretic inner bound to the achievable rate region in terms of {computable single-letter} information quantities (see Theorem \ref{thm:distPurity}). 

Toward the development of the results, we encounter two main challenges. The first challenge is in the compression of the joint measurements.
Since the classical communication allowed by the protocol is limited, the joint measurements, that Alice and Bob employ, are required to be compressed. Although a distributed measurement compression protocol for compressing a joint measurement have been developed earlier \cite{atif2021faithful}, one cannot directly use this protocol as a complete black box. The reason for this is that the measurement compression protocol also requires additional common randomness as a resource which the current purity distillation protocol does not allow. One may argue that derandomization or similar techniques could be used to remove common randomness constraint, but note that once the protocol is used a black box, derandomization techniques cannot remove the common randomness constraints. The authors in \cite{krovi2007local} has applied derandomization to eliminate common randomness, however, to the best of our knowledge, it fails to achieve the objective as one of their bounds (after \cite[Eq. 30]{krovi2007local}) still require additional common randomness. Apart from this, the measurement compression protocols provided in \cite{winter,wilde_e,atif2021faithful,atif2021distributed} shows the ``faithfulness'' of the post-measurement state of the reference along with the classical-quantum register storing the measurement outcome. These protocols remain unconcerned about the post-measurement state of the system on which the measurement is performed. However, in the current problem the closeness of the latter is needed. To the best of our knowledge, the authors in \cite{krovi2007local} do not make this distinction, and directly employ the result of \cite{winter}. To overcome this, we identify appropriate purifications of the post-measurement reference states and argue an existence of a collection of unitary operations achieving the latter (see Lemma \ref{lem:postMeasuredCloseness} for more details).

The second major challenge is that after the application of the compressed measurement, the states across the three parties are not necessary separable. This is because a compressed measurement is usually not a  ``sharp'' rank-one measurement. In \cite{devetak2005distillation} rank-one measurements are employed which makes the states separable and hence eases the analysis. In \cite{krovi2007local}, while using compressed measurements, the authors fail to justify the separability. To handle this, we develop a technique (see Lemma \ref{lem:Separate}) and employ it in our proof.
Lastly, as another contribution, we  incorporate the  asymptotic algebraic structured coding techniques
and provide a unified approach in characterizing the performance limits (see Def.~\ref{def:RateRegion}).


\section{Preliminaries}
\label{sec:prelim}
\noindent \textbf{Notation:} We supplement the notation in \cite{Wilde_book} with the following. Given any natural number $M$, let the finite set $\{1, 2, \cdots, M\}$ be denoted by $[1,M]$. 
Let $I$ denote the identity operator. 
Given a POVM $M\deq \{\Lambda^A_x\}_{x \in \mathcal{X}}$ acting on  $\rho$, the post-measurement state of the reference together with the classical outputs is represented by $     (\text{id} \tensor M)(\Psi^\rho_{RA})\deq \sum_{x\in \mathcal{X}} \ketbra{x}\tensor \tr_{A}\{(I^R \tensor \Lambda_x^A ) \Psi^\rho_{RA} \}.$ Let $\kappa(\rho^A)$ denote the asymptotic purity distillable  by  local purity concentration protocols from $\rho^A$ \cite{oppenheim2002thermodynamical,horodecki2003local}. We know $\kappa(\rho^A) = \log{\dim(\Cal{H}_A)} - S(\rho^A)$.

\begin{definition}[Faithful simulation \cite{wilde_e}]\label{def:faith-sim}
	Given a POVM ${M}\deq \{\Lambda_x\}_{x\in \mathcal{X}}$ acting on a $\mathcal{H}$ and a density operator $\rho\in \mathcal{D}(\mathcal{H})$, a sub-POVM $\tilde{M}\deq \{\tilde{\Lambda}_{x}\}_{x\in \mathcal{X}}$ acting on $\mathcal{H}$ is said to be $\epsilon$-faithful to $M$ with respect to $\rho$, for $\epsilon > 0$, if the following holds: 
	\begin{equation}\label{eq:faithful-sim-cond-1_2}
	\sum_{x\in \mathcal{X}} \Big\|\sqrt{\rho} (\Lambda_{x}-\tilde{\Lambda}_{x}) \sqrt{\rho}\Big\|_1+\tr\{(I-\sum_{x} \tilde{\Lambda}_{x})\rho\}   \leq \epsilon.
	\end{equation}
\end{definition}


\section{Distributed Purity Distillation}
In the following we describe the problem statement.
Let $\rho^{ABC}$ be a density operator acting on  $\mathcal{H}_A\tensor \mathcal{H}_B \tensor \mathcal{H}_C$.
Consider two measurements $M_A$ and $M_B$ on sub-systems $A$ and $B$, respectively. 
Imagine that we have three parties, named Alice, Bob and Charlie, trying to distill local purity from the noisy joint state $\rho^{ABC}$.
The resources available to these parties are 
(i) the classical communication links of specified rates between Alice and Charlie, and Bob and Charlie, modelled as a multiple-access dephasing channel, and (ii) an additional triple of pure catalytic quantum systems $A_C$, $B_C$ and $C_C$ available to Alice, Bob and Charlie, respectively.
Given the distributed nature of the problem, no communication is possible between Alice and Bob. 
The problem is formally defined in the following.

\begin{definition}\label{def:protocolDist}
	For a given finite set $\mathcal{Z}$, and a Hilbert space $\mathcal{H}_{A}\tensor \mathcal{H}_B\tensor \mathcal{H}_C$, a distributed purity distillation protocol with parameters $(n,\Theta_1,\Theta_2,\kappa_1, \kappa_2,\kappa_3, \iota_1,\iota_2,\iota_3)$ 
	is characterized by 
	
	\begin{enumerate}
	    \item a unitary operation on Alice's system $U_A \colon \CalH_A^{\tensor n} \tensor \CalH_{A_C} \rightarrow \CalH_{A_p}\tensor \CalH_{X_1}\tensor \CalH_{A_g}$, with $\dim(\CalH_{A_p})  = \kappa_1$, $\dim(\CalH_{A_C})  = \iota_1$, and $\dim(\CalH_{X_1})  = \Theta_1$.
	    \item a unitary operation on Bob's system $U_B \colon \CalH_B^{\tensor n} \tensor \CalH_{B_C} \rightarrow \CalH_{B_p}\tensor \CalH_{X_2}\tensor \CalH_{B_g}$, with $\dim(\CalH_{B_p})  = \kappa_2$, $\dim(\CalH_{B_C})  = \iota_2$, and $\dim(\CalH_{X_2})  = \Theta_2$.
	    \item a multiple access dephasing channel $\CalN \colon \CalH_{X_1} \tensor\CalH_{X_2} \rightarrow \CalH_{X_1} \tensor\CalH_{X_2}$.
	    \item a unitary operation on Charlie's system $U_C \colon \CalH_B^{\tensor n}  \tensor \CalH_{C_C} \tensor \CalH_{X_1} \tensor\CalH_{X_2}\rightarrow \CalH_{C_p}\tensor \CalH_{C_g}$, with $\dim(\CalH_{C_C})  = \iota_3$ and $\dim(\CalH_{C_p})  = \kappa_3$.
	\end{enumerate}

	%
	%
	%
	%
	%
\end{definition}
\begin{definition}\label{def:achievable}
	Given  a quantum state $\rho^{ABC}\in \mathcal{D}(\mathcal{H}_{A}\tensor \mathcal{H}_B \tensor \CalC)$, a triple $(P,R_1,R_2)$ is said to be achievable, if for all $\epsilon>0$ and for all sufficiently large $n$, there exists a distributed purity distillation protocol with parameters  $(n, \Theta_1, \Theta_2, \kappa_1, \kappa_2,\kappa_3, \iota_1,\iota_2,\iota_3)$ such that
	\begin{align}
	    G \deq \|\xi^{A_pB_pC_p} - \ketbra{0}^{A_p}&\tensor \ketbra{0}^{B_p}\tensor\ketbra{0}^{C_p}\|_1 \leq \epsilon, \nonumber \\
	    \frac{1}{n} \log_2 \Theta_i \leq R_i+\epsilon \colon  i\in [2],&\quad 
	    \frac{1}{n} \sum_{i\in[3]}\left(\log_2 \kappa_i - \log_2\iota_i\right) \leq P+\epsilon, \nonumber
	\end{align} 
	where $\ket{\xi} \deq U_C \Cal{N} U_B U_A |{\Psi^{\tensor n}_{\rho}}\ral^{ABCR}$, and $|{\Psi_\rho^{\tensor n}}\ral^{ABCR}$ is a purification of $(\rho^{ABC})^{\tensor n}$.
	The set of all achievable triples $( P,R_1,R_2)$ is called the achievable rate region. 
\end{definition}

Given a POVM $M\deq \{\Lambda^A_x\}_{x \in \mathcal{X}}$ acting on  $\rho$, the post-measurement state of the reference together with the classical outputs is represented by $     (\text{id} \tensor M)(\Psi^\rho_{RA})\deq \sum_{x\in \mathcal{X}} \ketbra{x}\tensor \tr_{A}\{(I^R \tensor \Lambda_x^A ) \Psi^\rho_{RA} \}.$
\begin{definition}\label{def:RateRegion}
		Consider a quantum state $\rho^{ABC}\in \mathcal{D}(\mathcal{H}_{A}\tensor \mathcal{H}_B \tensor \mathcal{H}_C)$, and a POVM
		$M_{AB}=\bar{M}_{A}\tensor \bar{M}_B$ acting on  $\mathcal{H}_{A}\tensor \mathcal{H}_B$ where $\bar{M}_A=\{\bar{\Lambda}^A_{s}\}_{s \in \mathcal{S}}$ and $\bar{M}_B=\{ \bar{\Lambda}^B_{t}\}_{t \in \mathcal{T}}$. Define the auxiliary states
\begin{align*}
    \sigma_{1}^{RBCS} &\deq (\emph{id}_R\tensor \bar{M}_{A} \tensor \emph{id}_{BC}) ( \Psi_{\rho}^{R A BC}), \quad
    \sigma_{2}^{RACT} \deq (\emph{id}_R\tensor \emph{\id}_{AC} \tensor \bar{M}_{B}) ( \Psi_{\rho}^{R A BC}),
     \quad \text{and} \nonumber \\
    & \sigma_3^{RST}   \deq \sum_{s,t}\sqrt{\rho^{AB}}\left (\bar{\Lambda}^A_{s}\tensor \bar{\Lambda}^B_{t}\right )\sqrt{\rho^{AB}} \tensor \ketbra{s}\tensor\ketbra{t},
\end{align*}
for some orthonormal sets $\{\ket{s}\}_{s\in\mathcal{S}}$ and $ \{\ket{t}\}_{t\in\mathcal{T}}$,
where $\Psi_{\rho}^{R A BC}$ is a purification of $\rho^{ABC}$. Let $\CalR_b(\rho^{ABC},M_{AB})$ be defined as the set of all pairs $(R_1,R_2)$ such that there exists a prime finite field $\mathbb{F}_p$, for a prime $p$, and a pair of mappings $f_S:\mathcal{S} \rightarrow \mathbb{F}_p$ and 
$f_T:\mathcal{T} \rightarrow \mathbb{F}_p$, 
yielding $U=f_S(S)$, $V=f_T(T)$, and either $ W=U+V $ (with respect to $\FF_p$)
or $ W = (U,V)$,
and the following inequalities are satisfied:
		\begin{subequations}
			\begin{align}
			R_1 &\geq I(U;RBC)_{\sigma_1}+I^+_b(W;V)_{\sigma_3} - I_b(U;V)_{\sigma_3},\nonumber\\
			R_2 &\geq I(V;RAC)_{\sigma_2}+I^+_b(W;U)_{\sigma_3}-I_b(U;V)_{\sigma_3},\nonumber\\
			{R_1+R_2} &\geq I(U;RBC)_{\sigma_1}+I(V;RAC)_{\sigma_2}-I_b(U;V)_{\sigma_3}+I^+_b(W;U)_{\sigma_3} +I^+_b(W;V)_{\sigma_3} \!- I^+_b(U;V)_{\sigma_3},\nonumber
			\end{align}
		\end{subequations}
     where  $I_b(\cdot)_{\sigma} = b\times I(\cdot)_{\sigma}$, and
    $I^+_b(W;U)_{\sigma_3}=I_b(W;U)_{\sigma_3}$, $I^+_b(W;V)_{\sigma_3}=I_b(W;V)_{\sigma_3}$, 
    $I^+_b(U;V)_{\sigma_3}=I_b(U;V)_{\sigma_3}$ if $W=U + V$, and 
    $I^+_b(W;U)_{\sigma_3}=I^+_b(W;V)_{\sigma_3}=I^+_b(U;V)_{\sigma_3}=0$ if $W=(U,V)$.
\end{definition}

\begin{theorem}\label{thm:distPurity}
    Given a quantum state $\rho^{ABC} \in \CalD(\CalH_A\tensor \CalH_B\tensor\CalH_C)$, 
      a triple $(R_1,R_2,P)$ is achievable if there exists a POVM $M_{AB}=\bar{M}_A\tensor \bar{M}_B$ acting on $\mathcal{H}_{A}\tensor \mathcal{H}_B$ with POVMs $\bar{M}_A=\{{\Lambda}^A_{s}\}_{s \in \mathcal{S}}$ and $\bar{M}_B=\{ {\Lambda}^B_{t}\}_{t \in \mathcal{T}}$ $\mathcal{H}_{A}\tensor \mathcal{H}_B$ and a real number $b\in [0,1]$ such that the following holds: 
    \begin{align}
    P & \leq \kappa(\rho_A) + \kappa(\rho_B) + \kappa(\rho_C) + I(C;W)_{\sigma} - I_b(U;V)_{\sigma_3}+I^+_b(W;U)_{\sigma_3} +I^+_b(W;V)_{\sigma_3} - I^+_b(U;V)_{\sigma_3}, \nonumber
    \end{align} and $(R_1,R_2) \in \CalR_b(\rho^{ABC},M_{AB})$,
    where $$\sigma^{RCST} \deq (\emph{id}_R\tensor \emph{\id}_{C} \tensor \bar{M}_{A}\tensor \bar{M}_{B}) ( \Psi_{\rho}^{R A BC}).$$
\end{theorem}
\begin{proof}
The proof is provided in Section \ref{sec:proofofTheoremDist}.
\end{proof}
\begin{definition}
        Given a quantum state $\rho^{ABC} \in \CalD(\CalH_A\tensor \CalH_B\tensor\CalH_C)$, and a dephasing channel with communication links of rates $R_1$ and $R_2$ define $1$-way distillable distributed local purity  $\kappa_{\rightarrow}(\rho^{ABC},R_1,R_2)$ as the supremum of the sum of all the locally distillable purity.
\end{definition}

\begin{corollary}
Given a quantum state $\rho^{ABC} \in \CalD(\CalH_A\tensor \CalH_B\tensor\CalH_C)$, let
\begin{align}
    \kappa^I_{\rightarrow}(\rho^{ABC},R_1,R_2) & \deq  \kappa(\rho^A) + \kappa(\rho^B) + \kappa(\rho^C)+ P^D_\rightarrow(\rho^{ABC},R_1,R_2), \nonumber\\
    P^D_\rightarrow(\rho^{ABC},R_1,R_2)& \deq \frac{1}{n}\lim_{n \rightarrow \infty}  \bar{P}^D_\rightarrow((\rho^{ABC})^{\tensor n},nR_1,nR_2), \nonumber\\
    \bar{P}^D_\rightarrow(\rho^{ABC},R_1,R_2)  & \deq \!\!\!\max_{M_{AB}, b\in[0,1]}\{I(C;W)_\sigma - I_b(U;V)_{\sigma} \nonumber \\ 
    & \hspace{0.2in}+I^+_b(W;U)_{\sigma_3} +I^+_b(W;V)_{\sigma_3} - I^+_b(U;V)_{\sigma_3} :(R_1,R_2) \in \CalR_b(\rho^{ABC},M_{AB}) \}. \nonumber
\end{align}
 With the above definitions, we have  $\kappa^I_{\rightarrow}(\rho^{ABC},R_1,R_2) \leq \kappa_{\rightarrow}(\rho^{ABC},R_1,R_2)$. In other words, for any communication rates $(R_1,R_2) $, $\kappa^I_{\rightarrow}(\rho^{ABC},R_1,R_2)$ amount of purity can be jointly distilled from the three parties using the protocol defined in Def. \ref{def:protocolDist}. 
\end{corollary}
\begin{proof}
The proof follows from Theorem \ref{thm:distPurity} and regularization.
\end{proof}

\section{Proof of Theorem \ref{thm:distPurity}}
\label{sec:proofofTheoremDist}
Observe that the theorem involves two different cases of $W$, one being equal to the sum $U+ V$, and another being the pair $(U,V)$.
We provide a complete proof for the latter case here. The proof of the former follows by employing the coding strategy from \cite[Theorem 2]{atif2021distributed} and performing a similar analysis as below.

The proof is mainly composed of two parts. In the first part, we construct a protocol by developing all the actions of the three parties, and describe them as unitary evolution (as these are the only actions allowed by the protocol, Def. \ref{def:protocolDist}).  Simultaneously, we also provide necessary lemmas needed for the next part. The second part deals with characterizing the action of the developed unitary operators on the shared quantum state $\rho^{ABC}$ and then bounding the error between the final state and the desired pure state. Since our result is derived for a bounded communication channel, we start by approximating the measurements to achieve a decreased outcome set, while preserving the statistics of the measurement. 
\subsection{Approximation of the measurement $M_A \tensor M_B$}
We start by generating the canonical ensembles corresponding to ${M}_{A}$ and ${M}_B$, defined as
\begin{align}\label{eq:dist_canonicalEnsemble}
\lambda^A_u \deq \tr\{{\Lambda}^A_u \rho^A\}, &\quad \lambda^B_v \deq \tr\{{\Lambda}^B_v \rho^B\}, \quad
 \lambda^{AB}_{uv} \deq \tr\{({\Lambda}^A_u \tensor \bar{\Lambda}^B_v) \rho^{AB}\}, \quad \text{and}\nonumber \\
\hat{\rho}^A_u \deq \frac{1}{\lambda^A_u}\sqrt{\rho^A}  {\Lambda}^A_u \sqrt{\rho^A}, &\quad \hat{\rho}^B_v \deq \frac{1}{\lambda^B_v}\sqrt{\rho^B}{\Lambda}^B_v \sqrt{\rho^B}, \quad 
 \hat{\rho}^{AB}_{uv} \deq \frac{1}{\lambda^{AB}_{uv}}\sqrt{\rho^{AB}} ({\Lambda}^A_u \tensor {\Lambda}^B_v) \sqrt{\rho^{AB}}.
\end{align}

Let $\PiA$ and $\PiB$ denote the $\delta$-typical projectors (as in \cite[Def. 15.1.3]{Wilde_book}) for marginal density operators $\rho^A$ and $\rho^B$, respectively.
Also, for any $u^n\in \mathcal{U}^n$ and $v^n \in \mathcal{V}^n$, let $\PiuA$ and $\PivB$ denote the strong conditional typical projectors (as in \cite[Def. 15.2.4]{Wilde_book}) for the canonical ensembles $\{\lambda^A_u, \hat{\rho}^A_u\}$ and $\{\lambda^B_v, \hat{\rho}^B_v\}$, respectively.

For each $u^n\in \TDeltan(U)$ and $v^n \in \TDeltan(V)$ define 
\begin{equation*}
\rhotilduA \deq \PiA \PiuA \rhohatuA \PiuA \PiA, \quad 
\rhotildvB \deq \PiB \PivB \rhohatvB \PivB \PiB,
\end{equation*}
and $\rhotilduA = 0,  $ and $\rhotildvB = 0$ for $u^n \notin \TDeltan(U)$ and $v^n \notin \TDeltan(V)$, respectively, with $\rhohatuA \deq \bigotimes_{i} \hat{\rho}^A_{u_i}$ and $\rhohatvB \deq \bigotimes_{i} \hat{\rho}^B_{v_i}$. 
Note that using the Gentle Measurement Lemma \cite{Wilde_book}, for any given $\epsilon \in (0,1)$, and sufficiently large $n$ and sufficiently small $\delta$, we have 
\begin{align}
    \|\rhohatuA - \rhotilduA\|_1 \leq \epsilon, \text{ and }\|\rhohatvB - \rhotildvB\|_1 \leq \epsilon, 
\end{align}
for all $u^n \in \TDelta(U)$ and $v^n \in \TDelta(V)$. Now we describe the random coding argument.
Randomly and independently select $2^{n\tilde{R}_1}$ and $2^{n\tilde{R}_2}$ sequences $(U^n(l), V^n(k))$ according to the pruned distributions, i.e.,
 \begin{align}\label{def:prunedDist}
     &\PP\left((U^{n,(\bar{\mu}_1)}(l), V^{n,(\bar{\mu}_2)}(k)) = (u^n,v^n)\right) = \left\{\begin{array}{cc}
          \dfrac{\lambdauA}{(1-\varepsilon)}\dfrac{\lambdavB}{(1-\varepsilon')} \quad & \mbox{for} \quad u^n \in \mathcal{T}_{\delta}^{(n)}(U), v^n \in \mathcal{T}_{\delta}^{(n)}(V)\\
           0 & \quad \mbox{otherwise}
     \end{array} \right. \!\!,
 \end{align} 
 where $\varepsilon = \sum_{u^n \in \mathcal{T}_{\delta}^{(n)}(U)}\lambdauA$ and $\varepsilon' = \sum_{v^n \in \mathcal{T}_{\delta}^{(n)}(V)}\lambdavB$. 
 Let $\mathcal{C}$ denote the codebook containing all  pairs of  codewords  $(U^{n}(l),V^{n}(k))$.

Construct operators
\begin{align}\label{eq:A_uB_v}
A_{u^n} &\deq  \gamma_{u^n} \bigg(\sqrt{\rho_{A}}^{-1}\rhotilduA\sqrt{\rho_{A}}^{-1}\bigg)\quad  \text{ and  } \quad
B_{v^n} \deq \zeta_{v^n} \bigg(\sqrt{\rho_{B}}^{-1}\rhotildvB\sqrt{\rho_{B}}^{-1}\bigg),
\end{align}
where 
\begin{align}\label{eq:gamma_mu}
 \gamma_{u^n}&\deq  \frac{1-\varepsilon}{1+\eta}2^{-n\tilde{R}_1}|\{l: U^{n}(l)=u^n\}|\quad  \text{ and  }
 \zeta_{v^n}\deq  \frac{1-\varepsilon'}{1+\eta}2^{-n\tilde{R}_2}|\{k: V^{n}(k)=v^n\}|, 
\end{align}
where $\eta \in (0,1)$ is a parameter that determines the probability of not obtaining sub-POVMs.
 Then construct $M_1^{( n)}$ and $M_2^{( n)}$ as in the following
\begin{align}
   M_1^{(n)}& \deq \{A_{u^n} \colon u^n \in  \mathcal{T}_{\delta}^{(n)}(U)\},
   M_2^{(n)} \deq \{B_{v^n} \colon v^n \in  \mathcal{T}_{\delta}^{(n)}(V)\}.
\label{eq:POVM_r1}
\end{align} 
We show later  that $M_1^{(n)}$ and $M_2^{(n)}$ form sub-POVMs, with high probability, 
 These collections $ {M}_1^{( n)}$ and $ {M}_2^{( n)}$ are completed using the operators $I - \sum_{u^n \in \mathcal{T}_{\delta}^{(n)}(U)}A_{u^n}$ and $I - \sum_{v^n\in \TDeltan(V)}B_{v^n}$, and these operators are associated with sequences $u^n_0$ and $v^n_0$, which are chosen arbitrarily from $\mathcal{U}^n\backslash\mathcal{T}_{\delta}^{(n)}(U) $ and $\mathcal{V}^n\backslash \mathcal{T}_{\delta}^{(n)}(V)$, respectively.
Let $\mathbbm{1}_{\{\mbox{sP-i}\}}$ denote the indicator random variable corresponding to the event that  $M_i^{(n)}$ form  sub-POVM for $i=1,2$. 
We use the trivial POVM $\{I\}$ in the case of the complementary event and associate it with $u_0^n$ and $v_0^n$ as the case maybe. In summary, the POVMs are given by $\{ \mathbbm{1}_{\{\mbox{sP-1}\}} A_{u^n}+(1- \mathbbm{1}_{\{\mbox{sP-1}\}})  \mathbbm{1}_{\{u^n=u^n_0\}} I \}_{u^n \in \mathcal{U}^n}$, and 
$\{ \mathbbm{1}_{\{\mbox{sP-2}\}} B_{v^n}+(1- \mathbbm{1}_{\{\mbox{sP-2}\}})  \mathbbm{1}_{\{v^n=v^n_0\}} I \}_{v^n \in \mathcal{V}^n}.$

Now, we intend to use the completions $[M_1^{( n, \bar{\mu}_1)}]$ and $[M_2^{( n, \bar{\mu}_2)}]$ in constructing the unitaries $U_A$ and $U_B$, as described in the protocol (Def. \ref{def:protocolDist}), for Alice and Bob, respectively.
Before concluding the discussion on the POVMs, we provide two lemmas which would be useful in the sequel.
The first lemma deals with bounding from below the probability that the constructed collection of operators indeed form a sub-POVM.
Toward this, observe that the collections of  approximating POVMs, $\{A_{U^n(l)}\}$ and $\{B_{V^n(k)}\}$, constructed in this work are identical to the ones employed in \cite{winter,wilde_e,atif2021faithful}, however, with one subtle difference: $A_{U^n(l)}$'s and $B_{V^n(k)}$'s do not have the outermost cut-off operator. Note that, it is only this cut-off operator, in the definition of $A_{U^n(l)}$'s and $B_{V^n(k)}$'s, which is constructed in a expected sense. Hence its absence allows us to maintain point-wise closeness of $\rhotilduA$ and $\rhohatuA$, (and similarly, $\rhotildvB$ and $\rhohatvB$,) without the need of expectation. This has profound implications. For instance, the result of Lemma \ref{lem:postMeasuredCloseness} is only possible after bypassing this operator. 

However, this detour does not allow us to employ the known operator Chernoff bound \cite[Lemma 17.3.1]{Wilde_book} directly. Hence, before providing the main lemma, we provide a slight variation of the former Chernoff bound as follows

\begin{lemma}[A new Operator Chernoff Bound]
\label{lem:ChernoffBound} Let $\{A_i\}_{ \in [N]}$ be a collection of $N$ IID random operators belonging to $\mathcal{L}(\mathcal{H})$ such that $0\leq A_i \leq 1 \quad \forall i \in [N]$. Let $\bar{A} \deq \frac{1}{N}\sum_{i=1}^N A_i$  and $A \deq \EE[\bar{A}]$. Suppose there exists an operator $\Pi$ such that $\Pi A \Pi \geq aI,$ for some $a \in (0,1)$, then for all $\eta \in (0,\min(\frac{1}{2},\frac{1-a}{a})$ we have
\begin{align}
    \PP\left((1-\eta)A \leq \bar{A} \leq (1+\eta)A \right) \geq 1-2\dim(\mathcal{H})\exp{-\frac{N\eta^2a}{4\ln{2}}}.
\end{align}
\end{lemma}
\begin{proof}
The proof is provided in Appendix \ref{proof:lem:ChernoffBound}.
\end{proof}
\begin{proposition}
\label{lem:Lemma for not sPOVM}
For any $\epsilon \in (0,1)$, any $\eta \in (0,1)$, any $\delta \in (0,1)$ sufficiently small, and any $n$ sufficiently large, we have 
$$\mathbb{E}\left[\mathbbm{1}_{\{\mbox{sP-1}\}}
\mathbbm{1}_{\{\mbox{sP-2}\}}\right]>1-\epsilon,$$ 
if $\tilde{R}_1> I(U;RB)_{\sigma_1}$ and $\tilde{R}_2> I(V;RA)_{\sigma_2}$, 
where $\sigma_1, \sigma_2$ are defined as in the statement of the theorem.
\end{proposition}
\begin{proof}
Observe that the collections  $\{A_{U^n(l)}\}$ and $\{B_{V^n(k)}\}$ satisfy all the hypotheses of the above Chernoff bound after identifying $\Pi$ as the cut-off operator employed in \cite{winter}. Now by following identical steps as in \cite{winter}, the result follows.
\end{proof}
The second lemma provides a unitary to show closeness of the post-measurement states obtained from approximating measurements and the actual measurements. Note that the faithful simulation results \cite{winter,wilde_e,atif2021faithful} show the closeness of states in the reference system, but the current result proves  the closeness of the post-measurement states. The main elements of the proof is in identifying appropriate purifications and using the Uhlmann's Theorem \cite{Wilde_book}.
The lemma is as follows.
\begin{lemma}\label{lem:postMeasuredCloseness}
Using the above definitions, for all $(u^n,v^n) \in \Cal{C}$ let
\begin{align}
     \sigmahatAu \deq \frac{(I^{E}\tensor\sqrt{\Lambda_{u^n}^A})\PsiRhoNket}{\sqrt{\lambdauA}} \;\text{ and }\;
    \sigmatildAu \deq \frac{(I^{E}\tensor\sqrt{A_{u^n}})\PsiRhoNket}{\sqrt{\gamma_{u^n}}}, \nonumber 
\end{align}
($\sigmahatBv$ and $\sigmatildBv$ defined analogously) where $E$ and $F$ denotes the system $BCR$ and $ACR$, respectively,  then for each $\settildeR{l}{1}$ and $\settildeR{k}{2}$ there exists a pair of  unitaries $U_{r}^{A}(l)$ and $U_{r}^{B}(k)$, such that 
\begin{align}
    F(\sigmahatAu,(I^E\!\tensor U_{r}^{A}(l))\!\sigmatildAu) &\geq\! \Big(1\!-\!\frac{1}{2}\|\rhohatuA - \rhotilduA\|_1\Big)^2, \; \mbox{ for } \; u^n = U^n(l) \; \mbox{and }\nonumber \\
    F(\sigmahatBv,(I^F\!\tensor U_{r}^{B}(k))\!\sigmatildBv) &\geq \!\Big(1\!-\!\frac{1}{2}\|\rhohatvB - \rhotildvB\|_1\Big)^2 \; \mbox{ for } \; v^n = V^n(k) \;. \nonumber
\end{align}
\end{lemma}

\begin{proof}
The proof is provided in Appendix \ref{proof:lem:postMeasuredCloseness}.
\end{proof}

\noindent We now move on to characterizing the unitaries $U_A$ and $U_B$.

\subsection{Action of Alice and Bob}
Using the approximating POVMs constructed above, as a  first unitary operation, Alice and Bob perform a coherent version of the approximating POVM. This is defined as
\begin{align}
\Unitary{M}{A} &\deq \sum_{l\in [2^{n\tilde{R}_1}]}\sqrt{A_{U^n(l)}}\tensor \ket{l}, \quad 
\Unitary{M}{B} \deq \sum_{k\in [2^{n\tilde{R}_2}]}\sqrt{B_{V^n(k)}}\tensor \ket{k}. \nonumber
\end{align}
Note that from now on, for the ease of notation, we use $\LambdalA, \LambdakB, \lambdalA, \lambdakB, A_l, B_k, \gamma_l,$ and $\zeta_k$ to denote the corresponding  $n-$letter objects constructed for the codewords $U^n(l)$ and $V^n(k)$, respectively.

Although the operators defined above are isometry operators, but with the help of additional catalyst qubits, these can be implemented as unitary operators.
Now, to extract purity from the states obtained after performing the measurements we employ the approach of \cite{krovi2007local}. More formally, 
we define the collection of unitaries
$\{U_p^A(l)\}_{\settildeR{l}{1}}$ and $\{U_p^B(k)\}_{\settildeR{k}{2}}$ as the unitaries that can extract purity for the collection of  states
$\{\hat\sigma^A_l\}_{\settildeR{l}{1}}$ and $\{\hat\sigma^B_k\}_{\settildeR{k}{2}}$, respectively. Note that since  $\hat\sigma^A_l$ and $\hat\sigma^B_k$ are product states, we use a type based construction (similar to one proposed in  \cite{krovi2007local}) in designing the unitary operators $U_p^A(l)$ and $U_p^B(k)$. 
However, note that since the approximating measurements are not rank-one operators, $U_p^A(l)$ and $U_p^B(k)$ will act on not necessarily separable states. This will not allow us to independently obtain the purity from the two parties, Alice and Bob. To address this, we use the fact that when extracting purity from Alice's state, the state is only slightly disturbed. 
More precisely, we provide the following lemma concerning the unitary operators $\{U_p^A(l)\}_{\settildeR{l}{1}}$ and $\{U_p^B(k)\}_{\settildeR{k}{2}}$. For $\settildeR{l}{1}$ and $\settildeR{k}{2}$, define the following collections of states:
\begin{align}
    |\hat\Psi_1(l)\ral & \deq \frac{(I^E\tensor \sqrt{\LambdalA})|\Psi_{\rho^{\tensor n}}\ral}{\sqrt{\lambdalA}},
    \quad 
    |\hat\Psi_2(k)\ral \deq \frac{(I^F\tensor \sqrt{\LambdakB})\Psi_{\rho^{\tensor n}}}{\sqrt{\lambdakB}}, \nonumber \\
    & |\hat\Psi_3(l,k)\ral \deq \frac{(I^{CR}\tensor \sqrt{\LambdalA\tensor \LambdakB})\Psi_{\rho^{\tensor n}}}{\sqrt{\lambdalkAB}}. \nonumber
\end{align}
\begin{lemma}\label{lem:purityProj}\label{lem:Separate}
Given the above definitions, for any given $\epsilon$ and sufficiently large $n$, and sufficiently small $\eta,\delta$, there exists three collections of projectors $\{\Pi^A_l\}_{\settildeR{l}{1}}$,  $\{\Pi^B_k\}_{\settildeR{k}{2}}$, and $\{\Pi^C_{l,k}\}_{\settildeR{l}{1},\settildeR{k}{2}}$, acting on Hilbert spaces $\CalH_A$, $\CalH_B$,  and $\CalH_C$, respectively, such that, for all $\settildeR{l}{1}$ and $\settildeR{k}{2}$, we have
\begin{align}
    F\left(U_p^A(l)|\hat\Psi_1(l)\ral, [(\Pi^A_l\tensor I^{RBC})|\hat\Psi_1(l)\ral]\tensor \ket{0}_{A_p} \right) &\geq 1- \epsilon, \nonumber\\
    F\left(U_p^B(k)|\hat\Psi_2(l)\ral, [(\Pi^B_k\tensor I^{RAC})|\hat\Psi_1(l)\ral]\tensor \ket{0}_{B_p} \right) &\geq 1- \epsilon, \nonumber\\
    F\left(U_p^C(l,k)|\hat\Psi_3(l,k)\ral, [(\Pi^C_{l,k}\tensor I^{RAB})|\hat\Psi_3(l,k)\ral]\tensor \ket{0}_{C_p} \right) &\geq 1- \epsilon, \nonumber
\end{align} and $\dim({\mathcal{H}_{A_p}}) \geq \log{\dim({\mathcal{H}_{A}})} - S(\hat{\rho}^A_{l}),\; \dim({\mathcal{H}_{B_p}}) \geq \log{\dim({\mathcal{H}_{B}})} - S(\hat{\rho}^B_{k}), \;\; $ and $ \;\dim({\mathcal{H}_{C_p}}) \geq \log{\dim({\mathcal{H}_{C}})} - S(\hat{\rho}^{AB}_{l,k}),\; $
with $\;\Tr{(\Pi^A_l\tensor I^{RBC})\hat\Psi_1(l)} \geq 1-\epsilon, \; \Tr{(\Pi^B_k\tensor I^{RAC})\hat\Psi_1(l) } \geq 1-\epsilon,$ and $ \Tr{(\Pi^C_{l,k}\tensor I^{RAB})\hat\Psi_3(l,k) } $ $\geq 1-\epsilon$.  
\end{lemma}
\begin{proof}
The proof follows from \cite[Lemma 1]{devetak2004distilling}.
\end{proof}
\noindent Now we characterize the complete action at Alice and Bob as
\begin{align}
    U_A &\deq \Unitary{P}{A}\Unitary{R}{A}\Unitary{M}{A} \;\text{ and }\; U_B \deq \Unitary{P}{B}\Unitary{R}{B}\Unitary{M}{B},
\end{align}
where $\Unitary{P}{A}$ and $\Unitary{R}{A}$ are controlled unitary operators defined as 
\begin{align}\label{def:U_P_AandB}
    \Unitary{P}{A} &\deq \sum_{\settildeR{l}{1}}\hspace{-5pt} U_p^A(l)\tensor \ketbra{l},  \quad 
    \Unitary{R}{A} \deq \sum_{\settildeR{l}{1}}\hspace{-5pt} U_r^A(l)\tensor \ketbra{l},
\end{align}
and similar is true for $\Unitary{P}{B}$ and $\Unitary{R}{B}$. This gives
\begin{align}
    U_A  = \sum_{\settildeR{l}{1}} U_p^A(l)U_r^A(l)\sqrt{A_{l}}\tensor \ket{l}, \quad 
    U_B  = \sum_{\settildeR{k}{2}} U_p^B(k)U_r^B(k)\sqrt{B_{k}}\tensor \ket{k}. \nonumber  
\end{align}
Finally, let 
\begin{align}
    &\ket{\Psi_1}^{ABCRLK}  \deq (I^{CR}\tensor U_A \tensor U_B) \ket{\Psi_{\rho^{\tensor n}}}^{ABCR}. \nonumber 
\end{align}
\subsection{Transmission over the Dephasing Channel $\Cal{N}$}
Before we proceed to employ the dephasing channel, observe that the classical registers created by the coherent measurement contains correlations across Alice and Bob. These correlations could be exploited which can further reduce the communication needed over the dephasing channel. For this, we employ the traditional binning operation. Begin by fixing the binning rates $(R_1, R_2)$, with $R_1 \leq \tilde{R}_1$ and $R_2 \leq \tilde{R}_2$.  For each sequence $u^n\in \mathcal{T}_{\delta}^{(n)}(U)$ assign an index  from $[1,2^{nR_1}]$ randomly and uniformly, such that the assignments for different sequences are done independently. Perform a similar random and independent assignment for all $v^n\in \mathcal{T}_{\delta}^{(n)}(V)$ with indices chosen from  $[1,2^{nR_2}]$. 
  For each $i\in [1,2^{nR_1}]$ and $j\in [1,2^{nR_2}]$, let $\mathcal{B}_1(i)$ and $\mathcal{B}_2(j)$ denote the $i^{th}$ and the  $j^{th}$ bins, respectively. More precisely, $\mathcal{B}_1(i)$ is the set of all $ u^n$ sequences with assigned index equal to $i$, and similar is $\mathcal{B}_2(j)$. {Also, note that the effect of the binning is in reducing the communication rates  from $(\tilde{R}_1, \tilde{R}_2)$ to $(R_1,R_2)$. }
  Moreover, let $\iota_1: \mathcal{T}_{\delta}^{(n)}(U) \rightarrow [1,2^{nR_1}]$, and 
  $\iota_2: \mathcal{T}_{\delta}^{(n)}(V) \rightarrow [1,2^{nR_2}]$,  denote the corresponding random binning functions. With this, we can denote $\ket{l}$ for $\settildeR{l}{1}$ as $\ket{l}_L = \ket{\iota_1(l)}_{L_1}\ket{\beta_U(l)}_{L_2}$ and similarly, $\ket{k}$ for $\settildeR{k}{2}$ as $\ket{k}_K = \ket{\iota_2(k)}_{K_1}\ket{\beta_V(k)}_{K_2}$\footnote{Note that $\iota_1(l) = \iota_1(U^n(l)) $, and similar holds for the functions $\iota_2, \beta_U, \beta_V$.}, where the functions $\beta_U$ and $\beta_V$ describe the remaining $\tilde{R}_1-R_1$ and $\tilde{R}_2-R_2$ qubits, respectively. 
  Now the qubits in the state $\ket{\iota_1(\cdot)}$ and $\ket{\iota_2(\cdot)}$ are sent over the multiple-access dephasing channel $\Cal{N}$, each requiring rates of $R_1$ and $R_2$ qubits, respectively. Let 
  $${\sigma^{ABCRLK}} \deq \Cal{N}({\Psi_1}^{ABCRLK}).$$
  
  With this, we move on to describing the action of Charlie.
  
  \subsection{Action of Charlie}
  \noindent Charlie begins by undoing the binning operation. For this, let
  \begin{align}
      D_{i,j} \deq & \big \{(l,k): (U^n(l),V^n(k)) \in \mathcal{T}_{\delta}^{(n)}(UV)  \text{ and }  (U^n(l), V^n(k)) \in \mathcal{B}_1(i)\times \mathcal{B}_2(j)\big \}.\nonumber
  \end{align}
 
For every $i\in [1,2^{nR_1}]$ and $j\in [1, 2^{nR_2}]$ define the function $F(i,j)=(l,k)$ if  $(l,k)$ is the only element of $D_{i,j}$; otherwise $F(i,j)=(0,0)$
Further, $F(i,j)=(0,0)$ for $i=0$ or $j=0$.
Using the qubits received from Alice and Bob, and the above definition of $F(i,j)$, Charlie aims at undoing the binning operations. This can be characterized as an isometric map $\Unitary{F}{C}: \mathcal{H}_{Y_1}\tensor\mathcal{H}_{Y_2} \rightarrow \mathcal{H}_{Y_1}\tensor\mathcal{H}_{Y_2}\tensor \mathcal{H}_{F}$ defined as
\begin{align}
 \Unitary{F}{C} \deq \sum_{i\in [2^{nR_1}]}\sum_{j\in [2^{nR_2}]}\ket{F(i,j)}\bra{i,j},
\end{align}
where $F()$ is such that $\dim(\mathcal{H}_F) = R_{tb} \deq \tilde{R}_1 - R_1 +  \tilde{R}_2 - R_2$. Note that, since binning decreased the total number of qubits transmitted by $R_{tb}$, to implement the above isometry, Charlie would need $R_{tb}$ number of additional catalytic qubits present in the pure state. 
As the protocol allows for the use of additional catalysts, as long as they are returned successfully, such an isometry can be implemented as a unitary.
\begin{remark}
  As will be shown in the sequel,  the error analysis gives an upper bound on $R_{tb}$. As this is only an upper bound, one can choose to not bin at the maximum rate and can save on the catalytic qubits needed. However, this would increase the communication rates by equivalent factors. This is modelled in the theorem statement using the real number $b\in [0,1]$.
  \end{remark}
After the complete identification of the measurement outcomes of Alice and Bob, Charlie now extracts the purity from her state, conditioned on these outcomes. For this, she develops a collection of unitary operations $\{U_p^C(l,k)\}_{\settildeR{l}{1},\settildeR{k}{2}}$, analogous to the earlier ones. Further, she constructs the controlled unitary $\Unitary{P}{C}$ defined as 
\begin{align}
    \Unitary{P}{C} \deq \sum_{\settildeR{l}{1}}\sum_{\settildeR{k}{2}}U_p^C(l,k)\tensor \ketbra{l,k}.
\end{align}
This characterizes Charlie's unitary as $U_C  = \Unitary{P}{C} \Unitary{F}{C}$, and gives
\[\xi^{ABCRLK} \deq (I\tensor\Unitary{P}{C} \Unitary{F}{C}) \sigma^{ABCRL_1K_1}(I\tensor \Unitary{P}{C} \Unitary{F}{C})^\dagger.\]

At this point, we have the characterized the actions of all the three parties as unitary operations.
The next step is to measure the distance between the obtained state and the desired pure state, and establish the $G$ can be made arbitrary small. 


\subsection{Analysis of Trace Distance}

We begin by defining the following.
\begin{align}
    \xi_1^{T_p} & \deq \tr_{RT_gLK}\{(I^{R}\tensor U_C \tensor U_A' \tensor U_B) {\Psi_{\rho^{\tensor n}}}(I^{R}\tensor U_C \tensor U_A' \tensor U_B)^\dagger\} \nonumber,
\end{align}
where 
\begin{align*}
   U_A' \deq \sum_{\settildeR{l}{1}} U_p^A(l)\sqrt{\frac{\gamma_l}{\lambdalA}}\sqrt{\LambdalA}\tensor \ket{l},
\end{align*}
 $T_p \deq \CalH_{A_p}\tensor \CalH_{B_p} \tensor \CalH_{C_p}$ and $T_g \deq \CalH_{A_g}\tensor \CalH_{B_g} \tensor \CalH_{C_g}$.
Also recall that,
\begin{align}
    \xi^{T_p} = \tr_{RT_gLK}\{(I^{R}\tensor U_C \tensor U_A \tensor U_B) {\Psi_{\rho^{\tensor n}}}(I^{R}\tensor U_C \tensor U_A \tensor U_B) ^\dagger\}
\end{align}
We first provide a proof for the case assuming the encoders do not perform any binning (i.e, $b=0$), and later incorporate the analysis for the setting when $b$ is non-zero. With this assumption, we define
\begin{align}
    \xi^{T_p}_b \deq \tr_{RT_gLK}\{(I^{R}\tensor U_P^C  \tensor U_A \tensor U_B) {\Psi_{\rho^{\tensor n}}}(I^{R}\tensor U_P^C  \tensor U_A \tensor U_B)^\dagger\},
\end{align}
where we have replaced $U_F^C$ with an identity transformation.
\par\textbf{Step 1: Closeness of $\xi_b^{T_p}$ and $\xi_1^{T_p}$:} As a first step, we show that $\xi_1^{T_p}$ can be made arbitrary close to $\xi_b^{T_p}$, in trace distance, for sufficiently large $n$. For this, define 
\begin{align*}
    V_1(l) \deq  I^{R}\tensor \Bigg(\sum_{\settildeR{k}{2}}U^C_p(l,k)\tensor \ketbra{k}\Bigg) \tensor U_p^{A}{(l)} \tensor U_B,
\end{align*} and consider the following:
\begin{align}
    \|\xi_1^{T_p} - \xi_b^{T_p}\|_1 & \leq \sum_{l}\bigg\|  V_1(l) \left( U_r^A(l)\sqrt{A_l}{\Psi_{\rho^{\tensor n}}}(U_r^A(l)\sqrt{A_l})^\dagger - \frac{\gamma_l}{\lambdalA} \sqrt{\LambdalA}{\Psi_{\rho^{\tensor n}}}\sqrt{\LambdalA}  \right)V_1(l)^\dagger\bigg\|_1 \nonumber \\
    & = \sum_{l}\gamma_l\bigg\| \frac{U_r^A(l)\sqrt{A_l}}{\sqrt{\gamma_l}}{\Psi_{\rho^{\tensor n}}}\frac{(U_r^A(l)\sqrt{A_l})^\dagger}{\sqrt{\gamma_l}} - \frac{\sqrt{\LambdalA}}{\sqrt{\lambdalA}} {\Psi_{\rho^{\tensor n}}}\frac{\sqrt{\LambdalA}}{\sqrt{\lambdalA}}  \bigg\|_1 \nonumber \\
    & = \sum_{l}\gamma_l\left\| U_r^A(l)\sigmatildAl(U_r^A(l))^\dagger-  \sigmahatAl \right\|_1 \nonumber \\
     & \leq 2\sum_{l}\gamma_l \sqrt{1-F\left( U_r^A(l)|{\sigmatildAl}\ral,|{\sigmahatAl}\ral \right)} \leq \delta_1, \nonumber 
\end{align}
where $\delta_1(\delta) \searrow 0$ as $\delta \searrow 0$, and the first inequality follows by using the monotonicity of trace distance, the triangle inequality and the definition of $V_1(l)$, the first equality follows by noting that $V_1(l)$ is a unitary for every $l$, the subsequent equality uses the fact that the trace distance is invariant with respect to an isometry \cite[Exercise 9.1.4]{Wilde_book}, and  definition of $\sigmahatAl$ and $\sigmatildAl$, the second inequality uses the relation between fidelity and trace distance (see \cite[Theorem 9.3.1]{Wilde_book}) and the last inequality follows by using Lemma \ref{lem:postMeasuredCloseness}, for sufficiently large $n$. 
With the above result, we now move on to the next step. For this define, 
\begin{align}
    \xi_2^{T_p} & \deq \tr_{RT_gLK}\{(I^{R}\tensor U_C \tensor U_A' \tensor U_B') {\Psi_{\rho^{\tensor n}}}(I^{R}\tensor U_C \tensor U_A' \tensor U_B')^\dagger\} \nonumber,
\end{align}
where
\begin{align*}
   U_B' \deq \sum_{\settildeR{k}{2}} U_p^B(k)\sqrt{\frac{\zeta_k}{\LambdakB}}\sqrt{\LambdakB}\tensor \ket{k}.
\end{align*}

\par\textbf{Step 2: Closeness of $\xi_1^{T_p}$ and $\xi_2^{T_p}$:} 
Recalling the definitions in \eqref{def:U_P_AandB}, define the unitary $V_2(k)$ for $\settildeR{k}{2}$, as 
\begin{align}
    V_2(k) \deq I^R \tensor \Bigg(\sum_{\settildeR{l}{1}}U^C_p(l,k)\tensor \ketbra{l}\Bigg) \tensor U_P^{A} \tensor U_p^{B}{(k)}.
\end{align}
Further, define the operators $V_3^A$ and $V_3^B$ as 
\begin{align}\label{def:V3AB}
    V_3^A \deq \sum_{\settildeR{l}{1}}\frac{\gamma_l}{\lambdalA}\sqrt{\LambdalA}\tensor \ket{l}, \quad \text{ and } \quad V_3^B \deq \sum_{\settildeR{k}{2}}\frac{\zeta_k}{\lambdakB}\sqrt{\LambdakB}\tensor \ket{k}.
\end{align}
Now, consider the following set of inequalities:
\begin{align}\label{eq:xi2_and_xi1}
    \|\xi_2^{T_p} - \xi_1^{T_p}\|_1 & \leq \left\| \sum_{k} V_2(k) V_3^A \left( \frac{\zeta_k}{\lambdakB}\sqrt{\LambdakB}\Psi_{\rho^{\tensor n}}\sqrt{\LambdakB} - U_r^B(k)\sqrt{B_k}\Psi_{\rho^{\tensor n}}(U_r^B(k)\sqrt{B_k})^\dagger  \right)(V_3^A)^{\dagger} V^\dagger_2(k) \right\|_1 \nonumber \\
    & \leq \sum_{k}\zeta_k \left\|  V_3^A \left( \frac{\sqrt{\LambdakB}}{\lambdakB}\Psi_{\rho^{\tensor n}}\frac{\sqrt{\LambdakB}}{\lambdakB} - \frac{U_r^B(k)\sqrt{B_k}}{\sqrt{\gamma_k}}\Psi_{\rho^{\tensor n}}\frac{(U_r^B(k)\sqrt{B_k})^\dagger}{\gamma_k})  \right)(V_3^A)^{\dagger} \right\|_1 \nonumber \\
    & = \sum_{k} \zeta_k  \left\|  V_3^A \sigmahatBk (V_3^A)^{\dagger}  - V_3^A U_r^B(k)\sigmatildBk U_r^B(k)^\dagger  (V_3^A)^{\dagger} \right\|_1  \nonumber \\
    & \leq 2\sum_{k} \zeta_k \sqrt{1 - F\left( V_3^A |\sigmahatBk\ral,V_3^A U_r^B(k)|\sigmatildBk \ral \right) }
\end{align}
where the first inequality follows by using the monotonicity of trace distance, and the definitions of $V_2(k)$ and $V_3^A$, the second inequality follows from the triangle inequality and the fact that the trace distance is invariant with respect to an isometry, and the equality follows from the definitions of $\sigmahatBk$ and $\sigmatildBk$. Further, we know that 
\begin{align}
    \EE_A\left[ (V_3^A)^\dagger V_3^A\right] &= \EE_A\left[\sum_{\settildeR{l}{1}}\frac{\gamma_l}{\lambdalA}\LambdalA \right]  = \EE_A\left[\sum_{u^n}\frac{\gamma_{u^n}}{\lambdauA}\LambdauA \right] \geq \frac{(1-\varepsilon)(1-\eta)}{(1+\eta)}I^A \label{eq:lowerboundV_3}
\end{align}
where $\EE_A$ denotes the expectation with respect to Alice's codebook generation process, and the last inequality follows from the proof of Proposition \ref{lem:Lemma for not sPOVM}, where we use the other inequality of the Chernoff bound, arguing that the expectation is close to identity from both sides.

This implies, 
\begin{align}
    \EE_A\left[F\left( V_3^A|\sigmahatBk\ral, V_3^A| U_r^B(k)| \sigmatildBk \ral\right)\right]  &= \EE_A \left[\lal \sigmahatBk |  (V_3^A)^\dagger V_3^A| U_r^B(k)| \sigmatildBk \ral \right]  \nonumber \\
    & \geq \frac{(1-\varepsilon)(1-\eta)}{(1+\eta)} \lal \sigmahatBk | U_r^B(k)| \sigmatildBk \ral \nonumber \\
    & = \frac{(1-\varepsilon)(1-\eta)}{(1+\eta)} F\left(|\sigmahatBk\ral,(I^F\tensor U_{r}^{B}(k))|\sigmatildBk\ral \right)   \geq 1-\delta_2 \label{eq:Fid_lowerbound}
\end{align}
where $\delta_2(\delta) \searrow 0$ as $\delta \searrow 0$, and the first inequality follows from \eqref{eq:lowerboundV_3} and the second follows from the result of Lemma \ref{lem:postMeasuredCloseness}.
Using the above inequality, and applying expectation to \eqref{eq:xi2_and_xi1}, we obtain
\begin{align}
    \EE_A\left[ \|\xi_2^{T_p} - \xi_1^{T_p}\|_1 \right] 
    & \leq  2\sum_{k} \zeta_k \sqrt{1 -\EE_A\left[ F\left( V_3 |\sigmahatBk\ral,V_3 U_r^B(k)|\sigmatildBk \ral \right) \right]} \leq 2\sqrt{\delta_2} \nonumber 
\end{align}
where the first inequality follows from the Jensen's inequality for square root function, and the second inequality follows from \eqref{eq:Fid_lowerbound}. 

Now we move on to the next step. Toward this, for $\settildeR{l}{1}$ and $\settildeR{k}{2}$, recall the definition of the following collections of states:
\begin{align}
    |\hat\Psi_1(l)\ral & \deq \frac{(I^E\tensor \sqrt{\LambdalA})|\Psi_{\rho^{\tensor n}}\ral}{\sqrt{\lambdalA}},
    \quad 
    |\hat\Psi_2(k)\ral \deq \frac{(I^F\tensor \sqrt{\LambdakB})\Psi_{\rho^{\tensor n}}}{\sqrt{\lambdakB}}, \nonumber \\
    & |\hat\Psi_3(l,k)\ral \deq \frac{(I^{CR}\tensor \sqrt{\LambdalA\tensor \LambdakB})\Psi_{\rho^{\tensor n}}}{\sqrt{\lambdalkAB}}. \nonumber
\end{align}
Now consider the following lemma.
Using the operators provided by the Lemma \ref{lem:purityProj}, define the projectors $\Pi^A, \Pi^B$, and $\Pi^C$ as
\begin{align}
    \Pi^A \deq \sum_l \Pi_l^A\tensor \ketbra{l}, \quad \Pi^B \deq \sum_k \Pi_k^B\tensor \ketbra{k}, \quad \Pi^C \deq \sum_{l,k} \Pi_{l,k}^C\tensor \ketbra{l,k} \nonumber
\end{align}
Considering the action of Alice in distilling purity, define 
\begin{align}
    \xi_{3}^{T_p} \deq  \tr_{RT_gKL}\left\{\left(I^R\tensor U_P^C \tensor U_B' \tensor \Pi^A V_3^A\right)\Psi_{\rho^{\tensor n}}\left(I^R\tensor U_P^C \tensor U_B' \tensor \Pi^A V_3^A\right)^\dagger\right\}\tensor \ketbra{0}_{A_p}. \nonumber
\end{align}


\par\textbf{Step 3: Closeness of $\xi_2^{T_p}$ and $\xi_3^{T_p}$:} 
Using the above definition, consider the following analysis:
\begin{align}\label{eq:xi3_and_xi2}
    \|\xi_2^{T_p} - \xi_3^{T_p}\|_1 \!
    &= \!\sum_l\!\gamma_l\left\| V_3^B \left( U_p^A(l)\hat\Psi_1(l)(U_p^A(l))^\dagger  \!- (\Pi^A_l\tensor I^{RBC})\hat\Psi_1(l)(\Pi^A_l\tensor I^{RBC}) \tensor\ketbra{0}_{A_p}\right) (V_3^B)^\dagger  \right\|_1 \nonumber \\ 
    &\leq 2\sum_l\gamma_l \sqrt{1 - F\left( V_3^B U_p^A(l)|\hat\Psi_1(l)\ral, V_3^B(\Pi^A_l\tensor I^{RBC})\hat\Psi_1(l)\tensor \ket{0}_{A_p}\right)}
\end{align}
where the above inequalities use similar set of arguments as in \ref{eq:xi2_and_xi1}. Now employing identical bounds as in \eqref{eq:lowerboundV_3} and \eqref{eq:Fid_lowerbound}, we obtain, for sufficiently large $n$, and sufficiently small $\delta,\eta$
\begin{align}
    \EE_B\left[ \|\xi_2^{T_p} - \xi_3^{T_p}\|_1\right] \leq 2\sqrt{\delta_2},
\end{align}
where $\EE_B$ denotes expectation over Bob's codebook and the inequality uses Jensen's inequality. This completes the current step.
Moving further, using the projectors defined in Lemma \ref{lem:purityProj}, we define $\xi_4^{T_p}$ as 
\begin{align}
    \xi_4^{T_p} \deq \tr_{RT_gLK}\left\{(\Pi^A)\left(I^R \tensor U_P^C \tensor V_3^A \tensor \Pi^B V_3^B\right)\Psi_{\rho^{\tensor n}}\left(I^R \tensor U_P^C \tensor V_3^A \tensor \Pi^B V_3^B\right)^\dagger(\Pi^A)\right\}\tensor \ketbra{0}_{A_pB_p}.
\end{align}

\par\textbf{Step 4: Closeness of $\xi_3^{T_p}$ and $\xi_4^{T_p}$:} We have  
\begin{align}
    \|\xi_3^{T_p} - \xi_4^{T_p}\|_1 &\leq \sum_k \zeta_k
    \left\|(\Pi^AV_3^A) \left( U_p^B(k) \hat\Psi_2(k) (U_p^B(k))^\dagger \right.\right. \nonumber \\
    & \left.\left. \hspace{1.5in}- \Pi^B_k \hat\Psi_2(k) \Pi^B_k \tensor \ketbra{0}_{B_p} \right) ((V_3^A)^\dagger\Pi^A)\tensor \ketbra{0}_{A_p}\right\|_1 \nonumber \\
    &\leq \sum_k \zeta_k
    \left\|V_3^A \left( U_p^B(k) \hat\Psi_2(k) (U_p^B(k))^\dagger - \Pi^B_k \hat\Psi_2(k) \Pi^B_k \tensor \ketbra{0}_{B_p} \right) ((V_3^A)^\dagger)\right\|_1,
\end{align}
where the second inequality follows by using the Holder's inequality and the fact that $\Pi^A$ is a projector,  (which implies $\|\Pi^A\|_{\infty} \leq 1$). Note that the right hand side of the above inequality is similar to the right hand side obtained in \eqref{eq:xi3_and_xi2}, and hence using the result of Lemma \ref{lem:purityProj} and similar arguments as in \eqref{eq:xi3_and_xi2}, we obtain, for sufficiently large $n$, and sufficiently small $\delta,\eta$
\begin{align}
     \EE_A\left[ \|\xi_3^{T_p} - \xi_4^{T_p}\|_1\right] \leq 2\sqrt{\delta_2}.
\end{align}

Now we define $\xi^{T_p}_5$ considering the action of Charlie as
\begin{align}
    \xi^{T_p}_5 &\deq \tr_{RT_gLK}\left\{(\Pi^A\tensor \Pi^B)\left(I^R \tensor \Pi^C \tensor V_3^A \tensor  V_3^B\right)\!\Psi_{\rho^{\tensor n}}\!\left(I^R \tensor \Pi^C \tensor V_3^A \tensor V_3^B\right)^\dagger\!(\Pi^A\tensor \Pi^B)\right\}\nonumber \\
    & \hspace{3in}\tensor \ketbra{0}_{A_pB_pC_p}. \nonumber
\end{align}
\par\textbf{Step 5: Closeness of $\xi_4^{T_p}$ and $\xi_5^{T_p}$:}  Using the result of Lemma \ref{lem:purityProj}, we proceed as follows:
\begin{align}
    \|\xi_4^{T_p} - \xi_5^{T_p}\|_1 & \leq \sum_{l,k}\frac{\gamma_l\zeta_k}{\lambdalA\lambdakB}\lambdalkAB\left\|(\Pi_l^A\tensor \Pi_k^B)\left( U_p^C(l,k) \hat\Psi_3(l,k) \left(U_p^C(l,k)\right)^\dagger \right.\right. \nonumber \\ 
    & \hspace{2in}\left.\left.- \Pi_{l,k}^C\Psi_3(l,k) \Pi_{l,k}^C\tensor \ketbra{0}_{C_p}\right) (\Pi_l^A\tensor \Pi_k^B)\tensor \ketbra{0}_{A_pB_p} \right\|_1 \nonumber\\
     & \leq \sum_{l,k}\frac{\gamma_l\zeta_k}{\lambdalA\lambdakB}\lambdalkAB\left\|\left( U_p^C(l,k) \hat\Psi_3(l,k) \left(U_p^C(l,k)\right)^\dagger - \Pi_{l,k}^C\Psi_3(l,k) \Pi_{l,k}^C\tensor \ketbra{0}_{C_p}\right) \right\|_1 \nonumber\\
    &\leq 2\sum_{l,k}\frac{\gamma_l\zeta_k}{\lambdalA\lambdakB}\lambdalkAB \sqrt{1- F\left(U_p^C(l,k) \hat|\Psi_3(l,k)\ral,\Pi_{l,k}^C|\Psi_3(l,k)\ral \right)} \leq 2\sum_{l,k}\frac{\gamma_l\zeta_k}{\lambdalA\lambdakB}\lambdalkAB \sqrt{\delta_2}.
\end{align}
This implies, $\EE_{A,B}\left[\|\xi_4^{T_p} - \xi_5^{T_p}\|_1\right] 2\sqrt{\delta_2}$, and hence for any given $\epsilon\in(0,1)$, and for sufficiently large $n$, and sufficiently small $\delta, \eta >0$,  $\EE_{A,B}\left[\|\xi_4^{T_p} - \xi_5^{T_p}\|_1\right]$ can be made arbitrary small.

Now as a final step, we consider the case when Alice and Bob chooses to bin their measurement outcomes before sending over the dephasing channel, i.e., the case when $b>0$. We term the error introduced by this process as the binning error.
\par{\textbf{Step 5: Closeness of $\xi_b^{T_p}$ and $\xi^{T_p}$}:} In this step, we bound the error that is introduced when Charlie tries to undo the binning operation by performing the unitary $U_F$. We show that Charlie will be successful if the rate at which binning is performed  is constrained by a non-trivial bound (to be obtained in Proposition \ref{prop:binning}), and hence the error involved in undoing the binning operation can be made arbitrary small, in an expected sense, for all sufficiently large $n$.

For $l \in \settildeR{l}{1}$ and $k \in \settildeR{k}{2}$, define $d(l,k)\deq F(i,j)$, such that $(U^n(l),V^n(k)) \in\mathcal{B}_1(i)\times \mathcal{B}_2(j)$. 
Note that $d(\cdot,\cdot)$ captures the overall effect of the binning followed by the decoding function $F$. Further, for $l \in \settildeR{l}{1}$ and $k \in \settildeR{k}{2}$, define 
\begin{align*}
    \tilde{\Psi}_{\rho^{\tensor n}} \deq \frac{(I^{RC}\tensor\sqrt{A_l\tensor B_k)}\Psi_{\rho^{\tensor n}}(I^{RC} \tensor \sqrt{A_l\tensor B_k})}{\gamma_l\zeta_k}.
\end{align*}
Using these definitions, we obtain
\begin{align}
    \|\xi^{T_p} - \xi_b^{T_p}\|_1 & \leq \sum_{l,k} \gamma_l \zeta_k \left\|(I^{R}  \tensor U_p^A(l)U_r^A(l) \tensor U_B(k)U^B_r(k))\left(  U_p^C(l,k){\tilde{\Psi}_{\rho^{\tensor n}}}(U_p^C(l,k))^\dagger \right. \right. \nonumber \\
    & \hspace{50pt} \left.\left.-  U_p^C(d(l,k)) \tilde{\Psi}_{\rho^{\tensor n}}\tensor U_p^C(d(l,k))^\dagger \right)(I^{R}  \tensor U_p^A(l)U_r^A(l) \tensor U_B(k)U^B_r(k))^\dagger \right\| \nonumber \\
    & = \sum_{l,k} \gamma_l \zeta_k \left\|  U_p^C(l,k){\tilde{\Psi}_{\rho^{\tensor n}}}(U_p^C(l,k))^\dagger -  U_p^C(d(l,k)) \tilde{\Psi}_{\rho^{\tensor n}}\tensor U_p^C(d(l,k))^\dagger  \right\| \nonumber \\
\end{align}
Now, consider the following proposition.
\begin{proposition}\label{prop:binning}
For any $\epsilon \in (0,1)$, and sufficiently large $n$ and sufficiently small $\eta,\delta > 0$, we have $\EE[\|\xi^{T_p} - \xi_b^{T_p}\|_1] \leq \epsilon$ if $\tilde{R}_1 - R_1 + \tilde{R_2} - R_2 \geq I(U;V)_{\sigma_3}$.
\end{proposition}
\begin{proof}
The proof is provided in Appendix \ref{proof:Prop:binning}.
\end{proof}

Finally, we complete the proof by combining the results from all the above steps in the following. Let $\11_{\{sP\}} \deq \11_{\{sP-1\}}\11_{\{sP-2\}}$. Using this, we have
\begin{align*}
 \|\xi^{T_p}\11_{\{sP\}} - \ketbra{0}^{T_p}\|_1
     & \leq \|\xi^{T_p} - \xi_b^{T_p}\|_1 +  \|\xi^{T_p}_b - \ketbra{0}^{T_p}\|_1 + \|\xi^{T_p}_b\|_1(1-\11_{\{sP\}})  \\
     & \leq  \|\xi^{T_p} - \xi_b^{T_p}\|_1 + 
     \|\xi^{T_p}_b - \xi_1^{T_p}\|_1 + 
     \|\xi^{T_p}_1 - \xi_2^{T_p}\|_1 + 
     \|\xi^{T_p}_2 - \xi_3^{T_p}\|_1 +
     \|\xi^{T_p}_3 - \xi_4^{T_p}\|_1 \\
     & \hspace{1in} + 
     \|\xi^{T_p}_4 - \xi_5^{T_p}\|_1 + 
     \|\xi^{T_p}_5 - \ketbra{0}^{T_p}\|_1 + \|\xi^{T_p}_b\|_1(1-\11_{\{sP\}})
\end{align*}
Taking expectation of the above inequality and using (i) the closeness of trace norm proved in each of the steps, and (ii) the result from Proposition \ref{lem:Lemma for not sPOVM}, we have the desired result. This completes the proof.

\appendices

\section{Proof of Lemmas}
\subsection{Proof of Lemma \ref{lem:ChernoffBound}}
\label{proof:lem:ChernoffBound}
Similar to \cite{Wilde_book}, we first make an additional asuumption in the Bernstein Trick \cite[Lemma 17.3.3]{Wilde_book} and  prove the following. 


Let  $\{X_i\}_{i\in [N]}$ be an IID positive semi-definite random collection of operators, represented by a generic random operator $X$. Then for any pair of operators $Y, \Pi \geq 0$, such that $\Pi Y \Pi \geq y \Pi$, and a positive real number $t$, the following inequality holds:
\begin{align}
    \PP\left( \sum_{i=1}^N X_i \nleq NY \right) \leq \dim{(\mathcal{H})}\|\EE\left[\exp{t\; \Pi(X-yI)\Pi}\right]\|^N_\infty
\end{align}
The proof of the above inequality is as follows:
\begin{align}
    \PP\left( \sum_{i=1}^N X_i \nleq NY \right) &= \PP\left( \sum_{i=1}^N (X_i-Y) \nleq 0 \right) \nonumber \\
    & = \PP\left( \sum_{i=1}^N T(X_i-Y)T^\dagger \nleq 0 \right) = \PP\left( \sum_{i=1}^N TX_iT^\dagger \nleq NTYT^\dagger \right), \nonumber
\end{align}
where $T = \sqrt{t}\Pi$. Since $\Pi Y \Pi \geq y\Pi$, we have 
\begin{align*}
    \PP\left( \sum_{i=1}^N TX_iT^\dagger \nleq NTYT^\dagger \right) &\leq \PP\left( \sum_{i=1}^N TX_iT^\dagger \nleq NT(yI)T^\dagger \right)\nonumber \\
    & = \PP\left( \sum_{i=1}^N T(X_i-yI)T^\dagger \nleq 0 \right)\nonumber\\
    & \leq \Tr\left\{\EE\left[\exp{\sum_{i=1}^N T(X_i-yI)T^\dagger}\right]\right\} \nonumber \\
    & \leq \dim{(\mathcal{H})}\|\EE\left[\exp{t\; \Pi(X-yI)\Pi}\right]\|^N_\infty,
\end{align*}
where the second inequality uses the Markov inequality and the last inequality follows from the arguments provided in \cite[(17.29) - (17.35)]{Wilde_book}.

Moving on, let $X' \deq \Pi X \Pi$,  which gives
\begin{align*}
    \|\EE\left[\exp{t\; \Pi(X-yI)\Pi}\right]\|^N_\infty = \|\EE\left[\exp{t\;X'} \exp{-ty\Pi}\right]\|^N_\infty
\end{align*}
Here, we make an additional assumption of $\EE[X] \leq mI$, for some $0 \leq m \leq 0 y.$ Observing that $X'$ and $\Pi$ commute, and using the inequality \cite[(17.41)]{Wilde_book}, we obtain
\begin{align*}
    \exp{tX'} - \Pi \leq X'(\exp{t}-1),
\end{align*}
giving us $\EE[\exp{tX'}] \leq \EE[X'](\exp{t}-1)+\Pi \leq (m\exp{t}+1-m)\Pi$.
Further, using the simplification \cite[(17.48)-(17.52)]{Wilde_book}, we get the relation
\begin{align*}
    \PP\left( \sum_{i=1}^N X_i \nleq NY \right) &\leq \dim{(\mathcal{H})}\|\EE\left[\exp{t\; \Pi(X-yI)\Pi}\right]\|^N_\infty \nonumber \\
    & \leq \dim{(\mathcal{H})}\big( (m\exp{t}+1-m)\exp{-ty} \big)^N \leq \dim{(\mathcal{H})}\exp{-ND(y||m)}
\end{align*}
where $D(\cdot||\cdot)$ denotes the Kullback–Leibler divergence \cite{2006EIT_CovTho}.
TO summarize, we have obtained the following: Given a collection of positive semi-definite random operators  $\{X_i\}_{i\in [N]}$ be an IID, such that $\EE[\frac{1}{N}\sum_i^N X_i] \leq mI $, for some $m > 0$ then for any pair of operators $Y, \Pi \geq 0$, such that $\Pi Y \Pi \geq y \Pi$, and $y > m$, the following inequality holds:
\begin{align*}
    \PP\left( \sum_{i=1}^N X_i \nleq NY \right) \leq \dim{(\mathcal{H})}\exp{-ND(y||m)}.
\end{align*}
Finally, the prove completes by identifying $X_i$ with $a A_m^{-\frac{1}{2}}A_iA_m^{-\frac{1}{2}}$ and $Y = (1+\eta)aI$ for some $\eta \in (0,\min(\frac{1}{2},\frac{1-a}{a}))$, and using the inequality $D((1+\eta)a||a) \geq \frac{\eta^2 a}{4\ln{2}}$.

\subsection{Proof of Lemma \ref{lem:postMeasuredCloseness}}
\label{proof:lem:postMeasuredCloseness}
We begin by observing the fact that $\sigmahatAu$ and $\sigmatildAu$ are purification of the states $\rhohatuA$ and $\rhotilduA$, respectively. More precisely,
\begin{align*}
    \rhohatuA = \Tr_{A}{|\hat{\sigma}_u^n\ral\lal\hat{\sigma}_u^n|^{AE}} \quad \mbox{ and } \quad \rhotilduA = \Tr_{A}{|\tilde{\sigma}_u^n\ral\lal\hat{\sigma}_u^n|^{AE}} 
\end{align*}
Also, note that the Hilbert space $\mathcal{H}_A^{\tensor}$, corresponding to the subsystem $A^n$ purifies the states $\rhohatuA$ and $\rhotilduA$. This implies, from Uhlmann's theorem, there exists a unitary $U_r(l)$ acting on the subsystem $A^n$, such that, for $u^n = U^n(l)$, we have
\begin{align*}
    F(\rhohatuA,\rhotilduA) = F(\sigmahatAu, U_r(l)\sigmatildAu).
\end{align*}
Finally, using the relation \cite[Theorem 9.3.1]{Wilde_book}, we obtain the desired result.
An identical analysis for $\sigmahatBv$ and $\sigmatildBv$ produces the second statement of the lemma.

\section{Proof of Propositions}
\subsection{Proof of Proposition \ref{prop:binning}}
\label{proof:Prop:binning}
We begin by defining  $\mathcal{J}$  as
\begin{align*}
    \mathcal{J} &\deq \left\{\exists (\tilde{l},\tilde{k},i,j): (U^n(l),V^n(k)) \in \mathcal{B}_1(i)\times\mathcal{B}_2(j),
    (U^n(\tilde{l}),V^n(\tilde{k})) \in \mathcal{B}_1(i)\times\mathcal{B}_2(j),\right. \nonumber \\
    & \left. \hspace{2in}(U^n(\tilde{l}),V^n(\tilde{k})) \in \TDelta(U,V) \right\}.
\end{align*}
Using this, consider the following simplification:
\begin{align}
    \|\xi^{T_p} - \xi_b^{T_p}\|_1  
    & \leq  \sum_{l,k} \gamma_l \zeta_k \11_{\{(l,k)\neq d(l,k)\}}2\underbrace{\|\tilde{\Psi}_{\rho^{\tensor n}}\|_1}_{\leq 1} \nonumber\\
    & \leq 2\sum_{l,k} \gamma_l \zeta_k \11_{\{ \mathcal{J}\}} = \frac{2}{2^{n(\tilde{R}_1+\tilde{R}_2)}}\sum_{u^n,v^n}\sum_{l,k} \11_{\{U^n(l) = u^n, V^n(k) = v^n\}} \11_{\{ \mathcal{J}\}}.\nonumber
\end{align}
Note that, for every $(u^n,v^n,l,k)$, we have
\begin{align*}
    \EE&\left[\11_{\{U^n(l) = u^n, V^n(k) = v^n\}} \11_{\{ \mathcal{J}\}}\right] \nonumber\\
    & \leq \sum_{(\tilde{u}^n,\tilde{v}^n)\in \TDelta(U,V)}\sum_{\tilde{l},\tilde{k}}\sum_{i,j}  \EE\left[\11\{U^n(l) = u^n, V^n(k) = v^n\}\11\{U^n(\tilde{l}) = \tilde{u}^n, V^n(\tilde{k}) = \tilde{v}^n\} \right.  \\
    & \hspace{1.5in}\left. \times\11{\{(u^n,v^n)\in\mathcal{B}_1(i)\times\mathcal{B}_2(j)\}}  \11{\{(\tilde{u}^n,\tilde{v}^n)\in\mathcal{B}_1(i)\times\mathcal{B}_2(j)\}}  \right] \nonumber \\
    &\leq  \frac{\lambdauA\lambdavB}{(1-\varepsilon)^2(1-\varepsilon')^2}
    2^{-n(I(U;V)-\delta_1)} \Big[2^{n(\tilde{R}_1-R_1)}2^{n(\tilde{R}_2-R_2)}+2^{n(\tilde{R}_1-R_1)}+2^{n(\tilde{R}_2-R_2)}\nonumber\\
    & \hspace{1in}+2^{-n(S(U)-\delta_1)} 2^{n\tilde{R}_1}2^{n(\tilde{R}_2-R_2)}+2^{-n(S(V)-\delta_1)} 2^{n\tilde{R}_2}2^{n(\tilde{R}_1-R_1)}\Big]\nonumber\\
    &\leq 5\frac{\lambdauA\lambdavB}{(1-\varepsilon)^2(1-\varepsilon')^2} 2^{-n(I(U;V)-2\delta_1)}2^{n(\tilde{R}_1-R_1)}2^{n(\tilde{R}_2-R_2)},
\end{align*}
where $\delta_1\ssearrow 0$ as $\delta\ssearrow 0$. {The first inequality follows from the union bound}. {The second inequality follows by evaluating the expectation of the indicator functions and the last inequality follows from the inequalities $\tilde{R}_1< S(U)$ and $\tilde{R}_2< S(V)$}. This implies,
\begin{align*}
    \EE\big[\|\xi^{T_p} - \xi_b^{T_p}\|_1  \big] \leq \frac{10}{(1-\varepsilon)^2(1-\varepsilon')^2} 2^{-n(I(U;V)-2\delta_1)}2^{n(\tilde{R}_1-R_1)}2^{n(\tilde{R}_2-R_2)},
\end{align*}
which completes the proof.

\newpage
\bibliographystyle{IEEEtran}

\bibliography{IEEEabrv,references,wisl_IC}

\end{document}